\documentclass[12pt]{iopart}

\usepackage{amsthm,graphicx}
\usepackage{graphicx}
\usepackage{colordvi}
\font\bb=msbm10 at 12pt

\newcommand{\mathbb}[1]{\hbox{\bb #1}} 
\def\softt{{\leavevmode\setbox1=\hbox{t}%
\hbox to \wd1{t\kern-0.6ex{\char039}\hss}}}

\newcommand{\R}{\mathbb{R}}

\newcommand{\OO}{\mathcal{O}}
\newcommand{\D}{\mathrm{d}}
\newcommand{\pd}{\partial}
\newtheorem{claim}{Claim}[section]
\newtheorem{theorem}[claim]{Theorem}

\newtheorem{lemma}[claim]{Lemma}

\begin{document}

\title[Strong $\delta'$ interaction on a planar loop]
{Spectral asymptotics of a strong $\delta'$ interaction on a planar loop}

\author{Pavel Exner}
\address{Doppler Institute for Mathematical Physics and Applied
Mathematics, \\ Czech Technical University in Prague,
B\v{r}ehov\'{a} 7, 11519 Prague, \\ and  Nuclear Physics Institute
ASCR, 25068 \v{R}e\v{z} near Prague, Czechia} \ead{exner@ujf.cas.cz}

\author{Michal Jex}
\address{Doppler Institute for Mathematical Physics and Applied
Mathematics, \\ Czech Technical University in Prague,
B\v{r}ehov\'{a} 7, 11519 Prague, \\ and Department of Physics,
Faculty of Nuclear Sciences and Physical Engineering, Czech
Technical University in Prague, B\v{r}ehov\'{a} 7, 11519 Prague,
Czechia} \ead{jexmicha@fjfi.cvut.cz}

\begin{abstract}
We consider a generalized Schr\"odinger operator in $L^2(\R^2)$ with an attractive strongly singular interaction of $\delta'$ type characterized by the coupling parameter $\beta>0$ and supported by a $C^4$-smooth closed curve $\Gamma$ of length $L$ without self-intersections. It is shown that in the strong coupling limit, $\beta\to 0_+$, the number of eigenvalues behaves as $\frac{2L}{\pi\beta} + \OO(|\ln\beta|)$, and furthermore, that the asymptotic behaviour of the $j$-th eigenvalue in the same limit is $-\frac{4}{\beta^2} +\mu_j+\OO(\beta|\ln\beta|)$, where $\mu_j$ is the $j$-th eigenvalue of the Schr\"odinger operator on $L^2(0,L)$ with periodic boundary conditions and the potential $-\frac14 \gamma^2$ where $\gamma$ is the signed curvature of $\Gamma$.
\end{abstract}

\maketitle

\section{Introduction}

Schr\"odinger operators with singular interactions supported by manifolds of a lower dimension have been studied for several decades starting from early works \cite{Ku, BT}. In recent years they attracted attention as a model of a quantum particle confined to sets of nontrivial geometry, a possible alternative to the usual quantum graphs \cite{BK} having two advantages over the latter. The first is that they lack the abundance of free parameters associated with the vertex coupling, the second, physically maybe more important, is that the confinement is not strict and a certain tunneling between parts of the graph is allowed. One usually speaks about `leaky' quantum graphs and describes them by Hamiltonians which can be formally written as $-\Delta - \alpha \delta(\cdot-\Gamma),\: \alpha>0,\,$ where $\Gamma$ is the support of the attractive singular interaction. A discussion of such operators and a survey of their properties can be found in \cite{Ex}.

One can think of the singular interaction as of a $\delta$ potential in the direction perpendicular to $\Gamma$, at least at the points where the manifold supporting the interaction is smooth. If the codimension of $\Gamma$ is one, however, there are other singular interactions which can be considered, a prime example being the one coming from the one-dimensional $\delta'$ interaction \cite{AGHH}, that is, operators which can be formally written as
\begin{equation*}
H=-\Delta-\beta^{-1}\delta'(\cdot-\Gamma)\,.
\end{equation*}
The formal expression has to be taken with a substantial grain of salt, of course, because in contrast to the $\delta$ interaction which can be approximated by naturally scaled regular potentials, the problem of approximating $\delta'$ is considerably more complicated --- see \cite{Se, CS, ENZ} and also \cite{CAZ, GH}. What is important for our present purpose, however, is that irrespective of the meaning of such an interaction, there is a mathematically sound way how to define the above operator through boundary conditions, and moreover, one can also specify it using the associated quadratic form \cite{BLL}.

Apart from the definition, one is naturally interested in spectral properties of such operators, in particular, in relation to the geometry of $\Gamma$. In the case of $\delta$-type singular interaction we know, for instance, that $\Gamma$ in the form of broken or bent line gives rise to a nontrivial discrete spectrum \cite{EI} and a similar result can be proven also for the $\delta'$-interaction \cite{BEL}. In this paper we want demonstrate another manifestation of the relation between eigenvalues of $H$ and the shape of $\Gamma$. It is inspired by the paper \cite{EY} in which it was shown how the eigenvalues coming from a $\delta$ interaction supported by a $\mathcal C^4$ Jordan curve $\Gamma$ behave in the strong-coupling regime, $\alpha\to \infty$, namely that after a renormalization consisting of subtracting the $\Gamma$-independent divergent term they are in the leading order given by the respective eigenvalue of a one-dimensional Schr\"odinger operator with a potential determined by the curvature of $\Gamma$.

Here we are going to show that in the $\delta'$ case, where the strong coupling limit is $\beta\to 0_+$, we have an analogous result, namely that the asymptotic expansion of the eigenvalues starts from a $\Gamma$-independent divergent term followed by the appropriate eigenvalues of a one-dimensional Schr\"odinger operator, the same one as in the $\delta$ case. We will be also able to derive an asymptotic expression for the number of eigenvalues dominated by a natural Weyl-type term. In the next section we state the problem properly and formulate the indicated results, the next two sections are devoted to the proofs. The technique is similar to that of \cite{EY}, however, the argument is slightly more complicated because the present form of the associate quadratic form does allow one to estimate the operator in question by operators with separated variables. In conclusion we shall comment briefly on possible extensions of the results.

\section{Formulation of the problem and main results}
\setcounter{equation}{0}

We consider a closed curve $\Gamma$ without self-intersections, conventionally parameterized by its arc length,
\begin{equation*}
\Gamma:\:[0,L]\rightarrow\mathbb{R}^2\,, \quad s\mapsto(\Gamma_1(s),\Gamma_2(s))\,,
\end{equation*}
with the component functions $\Gamma_1,\Gamma_2\in C^4(\mathbb R)$. The operator, we are interested in, can acts as the Laplacian outside the interaction support,
\begin{equation*}
(H_\beta\psi)(x)=-(\Delta\psi)(x)
\end{equation*}
for $x\in\mathbb{R}^2\setminus\Gamma$, and its domain is $\mathcal{D}(H_\beta)=\{\psi\in H^2(\mathbb{R}^2\setminus\Gamma) \mid \partial_{n_\Gamma}\psi(x)=\partial_{-n_\Gamma}\psi(x)=\psi'(x)|_{\Gamma},\, -\beta\psi'(x)|_{\Gamma}=\psi(x)|_{\partial_+\Gamma} -\psi(x)|_{\partial_-\Gamma}\}$, where $n_\Gamma$ is the normal to $\Gamma$, for definiteness supposed to be the outer one, and $\psi(x)|_{\partial_\pm\Gamma}$ are the appropriate traces of the function $\psi$. The quadratic form associated with this operator is well known \cite[Prop.~3.15]{BLL}. In order to write it, we employ the locally orthogonal curvilinear coordinates $(s,u)$ in the vicinity of the curve introduced in relation (\ref{curvilin}) below. With an abuse of notation we write the value of a function $\psi\in C(\R^2) \cap H^1(\mathbb{R}^2\setminus\Gamma)$ as $\psi(s,u)$; then we have
\begin{equation*}
h_\beta[\psi]= \|\nabla \psi\|^2 -\beta^{-1}\int_{\Gamma}|\psi(s,0_+)-\psi(s,0_-)|^2\,\D s\,.
\end{equation*}
To state our main theorem we introduce the following operator,
\begin{equation} \label{comparison}
S=-\frac{\partial^2}{\partial s^2}-\frac{1}{4}\gamma(s)^2\,,
\end{equation}
where $\gamma$ denotes the signed curvature of the loop, $\gamma(s) := (\Gamma''_1 \Gamma'_2 - \Gamma'_1 \Gamma''_2)(s)$. The domain of this operator is $\mathcal D(S)=\{\psi\in H^2(0,L)\mid \psi(0)=\psi(L),\,\psi'(0)=\psi'(L)\}$. We denote by $\mu_j$ the $j$-th eigenvalue of $S$ with the multiplicity taken into account.
\begin{theorem} \label{thm1}
One has $\sigma_\mathrm{ess}(H_\beta) = [0,\infty)$ and to any $n\in\mathbb{N}$ there is a $\beta_n>0$ such that
$$ 
\#\sigma_\mathrm{disc}(H_\beta)\geq n\quad\mathrm{holds \;\,for}\quad \beta\in(0,\beta_n)\,.
$$ 
For any such $\beta$ we denote by $\lambda_j(\beta)$ the $j$-th eigenvalue of $H_\beta$, again counted with its multiplicity. Then the asymptotic expansions
\begin{equation*}
\lambda_j(\beta)= -\frac{4}{\beta^2} +\mu_j+\mathcal{O}\big(\beta|\ln\beta|\big)\,, \quad j=1,\dots,n\,,
\end{equation*}
are valid in the limit $\beta\to 0_+$.
\end{theorem}

\medskip

\begin{theorem} \label{thm2}
The counting function $\beta\mapsto \#\sigma_\mathrm{disc}(H_\beta)$ admits the asymptotic expansion
\begin{equation*}
\#\sigma_\mathrm{disc}(H_\beta) = \frac{2L}{\pi\beta}+\mathcal{O}(|\ln\beta|)
\quad\mathrm{as}\;\; \beta\to 0_+\,.
\end{equation*}
\end{theorem}

\section{Proof of Theorem~\ref{thm1}}
\setcounter{equation}{0}

The essential spectrum of $H_\beta$ is found in \cite[Thm.~3.16]{BLL}. To prove the claim about the discrete one we need first a few auxiliary results. To begin with, we introduce locally orthogonal curvilinear coordinates $s$ and $u$ which allow us to write points $(x,y)$ in the vicinity of the curve as
\begin{equation} \label{curvilin}
(x,y)=\big(\Gamma_1(s)+u\Gamma'_2(s),\Gamma_2(s)-u\Gamma'_1(s)\big)\,.
\end{equation}
Since $\Gamma$ is supposed to be a $C^4$ smooth closed Jordan curve, it is not difficult to establish that the map (\ref{curvilin}) is injective for all $u$ small enough; for a detailed proof see \cite{EY}.

We choose a strip neighbourhood $\Omega_a:= \{x\in\R^2:\, \mathrm{dist\,}(x,\Gamma)<a\}$ of $\Gamma$ with $a$ small enough to ensure the injectivity and use bracketing to get a two-sided estimate of the operator $H_\beta$ by imposing Dirichlet and Neumann condition at the boundary of $\Omega_a$, i.e.
\begin{equation} \label{bracketing}
H_N(\beta)\leq H_\beta\leq H_D(\beta)\,,
\end{equation}
where both the estimating operators correspond to the same differential expression and  $\mathcal D(H_N(\beta))=\{\psi\in\mathcal D(H_\beta)\mid \partial_{u_+}\psi(s,a)=\partial_{u_-}\psi(s,-a)=0\}$ while the other is $\mathcal D(H_D(\beta))=\{\psi\in\mathcal D(H_\beta) \mid \psi(s,a)=\psi(s,-a)=0\}$. The operators $H_D(\beta)$ and $H_N(\beta)$ are obviously direct sums of operators corresponding to the parts of the plane separated by the boundary conditions, and since their parts referring to $\R^2 \setminus \overline\Omega_a$ are positive, we can neglect them when considering the discrete spectrum. The parts of $H_N(\beta)$ and $H_D(\beta)$ referring to the strip $\Omega_a$ are associated with the following quadratic forms,
\begin{eqnarray*}
h_{N,\beta}[f]=\|\nabla f\|^2 -\beta^{-1}\int_{\Gamma}|f(s,0_+)-f(s,0_-)|^2\, \D s\,, \\
h_{D,\beta}[g]=\|\nabla g\|^2 -\beta^{-1}\int_{\Gamma}|g(s,0_+)-g(s,0_-)|^2\, \D s\,,
\end{eqnarray*}
respectively, the former being defined on $H^1(\Omega_a\setminus\Gamma)$, the latter on $H_0^1(\Omega_a\setminus\Gamma)$. Our first task is to rewrite these forms in terms of the curvilinear coordinates $s$ and $u$.

\begin{lemma}
Quadratic forms $h_{N,\beta}$, $h_{D,\beta}$ are unitarily
equivalent to quadratic forms $q_{N,\beta}$ and $q_{D,\beta}$ which
can be written as
\begin{eqnarray*}
\hspace{-3em} q_D[f]= \|\frac{\partial_s f}{g}\|^2 +\|\partial_u f\|^2 +(f,Vf) -\beta^{-1}\int_{0}^{L}|f(s,0_+)-f(s,0_-)|^2\, \D s \\
\hspace{-.5em} + \frac12 \int_{0}^{L} \gamma(s) \big(|f(s,0_+)|^2-|f(s,0_-)|^2\big)\, \D s \\[.3em]
\hspace{-3em} q_N[g]= q_D[g] -\int_{0}^{L}\frac{\gamma(s)}{2(1+a\gamma(s))}|f(s,a)|^2\, \D s
+\int_{0}^{L}\frac{\gamma(s)}{2(1-a\gamma(s))}|f(s,-a)|^2\, \D s
\end{eqnarray*}
defined on $H_0^1((0,L) \times ((-a,0)\cup (0,a)))$ and $H^1((0,L) \times ((-a,0)\cup (0,a)))$, respectively, with periodic boundary conditions in the variable $s$. The geometrically induced potential in these formul{\ae} is given by  $V=\frac{u\gamma''}{2g^3} -\frac{5(u\gamma')^2}{4g^4} -\frac{\gamma^2}{4g^2}$ with $g(s):= 1+u\gamma(s)$.
\end{lemma}
\begin{proof}
Using the conventional shorthands, $\partial_s = \frac{\partial}{\partial s}$ etc., we express $\partial_s$ and $\partial_u$ as linear combinations of $\partial_1$ and $\partial_2$ with the coefficients $\partial_s x_1= \Gamma_1'+u\Gamma_2''$, $\partial_s x_2= \Gamma_2'-u\Gamma_1''$, $\partial_u x_1= \Gamma_2'$, and $\partial_u x_2= -\Gamma_1'$. Inverting these relations we get
\begin{equation*}
\partial_1 = g^{-1}\big(-\Gamma_1' \partial_s -(\Gamma_2'-u\Gamma_1'') \partial_u \big)\,, \quad
\partial_2 = g^{-1}\big(-\Gamma_2' \partial_s +(\Gamma_1'+u\Gamma_2'') \partial_u \big)\,,
\end{equation*}
where $g=(\Gamma_1'+u\Gamma_2'')\Gamma_1' +(\Gamma_2'-u\Gamma_1'')\Gamma_2' =\Gamma_2''\Gamma_1'-\Gamma_2'\Gamma_1'' =1+u\gamma$ because $\Gamma_1'^2+\Gamma_2'^2=1$ holds by assumption. The last relation gives $\Gamma_1'^2+\Gamma_2'^2=1$ which in turn implies $\gamma^2=\Gamma_1''^2+\Gamma_2''^2$. Using these identities we can check by a direct computation that
\begin{equation*}
q_{j,\beta}[Uf] =h_{j,\beta}[f]
\end{equation*}
with $(Uf)(s,u):= \sqrt{1+u\gamma(s)}\, f(x_1(s,u),x_2(s,u))$ holds for $j=D,N$ and all functions $f\in\mathcal{D}(h_{i,\beta})$, which proves the claim.
\end{proof}

The forms $q_{N,\beta}$ and $q_{D,\beta}$ are still not easy to handle and we are going to replace the estimate (\ref{bracketing}) by a cruder one in terms of following forms associated with operators. As for the upper bound, we introduce the quadratic form $q_{a,\beta}^{+}$ acting as
\begin{eqnarray*}
\lefteqn{q_{a,\beta}^{+}[f]=\|\pd_u f\|^2
+(1-a\gamma_+)^{-2}\|\pd_s f\|^2 +(f,V^{(+)}f)} \\ &&
-\beta^{-1}\int_{0}^{l}|f(s,0_+)-f(s,0_-)|^2 \D s +\frac12
\int_{0}^{l} \gamma(s) \big(|f(s,0_+)|^2-|f(s,0_-)|^2\big)\, \D s
\end{eqnarray*}
where $V^{(+)}:=\frac{a\gamma_+''}{2(1-a\gamma_+)^3}-\frac{\gamma^2}{4(1+a\gamma_+)^2}$ with $\gamma_+'':= (\gamma'')_+$ and the positive (negative) part given by the standard convention, $f_\pm:= \frac12 (|f|\pm f)$; we have neglected here the non-positive term $-\frac54 (u\gamma')^2 g^{-4}$. In contrast to the argument used in the $\delta$ interaction case \cite{EY} the operator $Q_{a,\beta}^{+}$ associated with this form does not have separated variables, however, one can write it as $Q_{a,\beta}^{+}=U_{a}^+\otimes I+\int_{[0,L)}^{\oplus} T_{a,\beta}^{+}(s)\,\D s$ and we are going to show that the spectrum of second part associated with the form
\begin{equation*}
t_{a,\beta}^{+}(s)[f]:= \|f'\|^2 -\frac{1}{\beta}\, |f(0_+)-f(0_-)|^2
+\frac12 \gamma(s) \big(|f(s,0_+)|^2-|f(s,0_-)|^2\big)
\end{equation*}
is independent of $s$. The operator itself acts as $T_{a,\beta}^{+}(s)f =-f''$ with the domain
\begin{eqnarray*}
\hspace{-2em} \lefteqn{\mathcal{D}(T_{a,\beta}^{+}(s))=
\big\{f\in H^2((-a,a)\setminus\{0\})\mid f(a)=f(-a)=0\,,} \\ &&
f'(0_-)=f'(0_+)= -\beta^{-1}(f(0_+)-f(0_-))+\frac12 \gamma(s)(f(0_+)+f(0_-)) \big\}
\end{eqnarray*}

\begin{lemma} \label{trans+}
The operators $T_{a,\beta}^{+}(s)$ has exactly one negative eigenvalue $t_+=-\kappa_+^2$ provided $\frac{a}{\beta}>2$ which is independent of $s$ and such that
\begin{equation*}
\kappa_+=\frac{2}{\beta} -\frac{4}{\beta}\, \e^{-4a/\beta}+
\mathcal{O} (\beta^{-1} \e^{-8a/\beta}) \quad\; \mathit{holds\;as}\quad \beta\to 0\,.
\end{equation*}
\end{lemma}
\begin{proof}
An eigenfunction corresponding to the eigenvalue $-\kappa^2$ and obeying the conditions $f(\pm a)=0$ and $f'(0_-)=f'(0_+)$ is, up to a multiplicative constant, equal to $\sinh(\kappa(x\mp a))$ for $\pm x\in(0,a)$. The function is odd, hence $f(0_-)=-f(0_+)$ and the $s$-dependent term does not influence the eigenvalue; the spectral condition is easily seen to be
\begin{equation}
\label{spec} \kappa= \frac{2}{\beta}\tanh(\kappa a)\,.
\end{equation}
We are interested in the asymptotic behaviour of the solution as $\beta\to 0_+$. Let us rewrite the condition as $\beta= \frac{2}{\kappa}\tanh(\kappa a)$; since the right-hand side is monotonous as a function of $\kappa$ it is clear that there is at most one eigenvalue and that this happens if $\beta<2a$. Furthermore, the right-hand side is less that $\frac{2}{\kappa}$ which means that $\kappa<2\beta^{-1}$ and the inequality turns to equality as $\beta\to 0$ and $\kappa\to\infty$. Next we employ Taylor expansion
\begin{equation*}
\frac{2}{\beta}\tanh(\kappa a) = \frac{2}{\beta}\,
\Big( 1 - 2\e^{-2\kappa a} + 2\e^{-4\kappa a} + \mathcal{O}(\e^{-6\kappa a}) \Big)\,,
\end{equation*}
and since $\kappa\to \frac{2}{\beta}$ as $\beta\to 0$, relation (\ref{spec}) yields the sought result.
\end{proof}

Next we estimate in a similar fashion the operator with Neumann boundary condition which we need to get a lower bound. To this aim we employ the quadratic form $q_{a,\beta}^{-}$ defined as
\begin{eqnarray*}
\lefteqn{q_{a,\beta}^{-}[f]=\|\pd_u f\|^2+(1+a\gamma_+)^{-2} \|\pd_s f\|^2 +(f,V^{(-)}f)} \\ &&
-\beta^{-1}\int_{0}^{l}|f(s,0_+)-f(s,0_-)|^2 \,\D s - \frac12 \int_{0}^{l}\gamma(s) \big(|f(s,0_+)|^2-|f(s,0_-)|^2\big) \,\D s \\ && -\gamma_+\int_{0}^{l}|f(s,a)|^2\,\D s-\gamma_+\int_{0}^{l}|f(s,-a)|^2\,\D s
\end{eqnarray*}
where $V^{(-)}=-\frac{a\gamma_+''}{2(1-a\gamma_+)^3} -\frac{5(a\gamma'_+)^2}{4(1-a\gamma_+)^4} -\frac{\gamma^2}{4(1-a\gamma_+)^2}$. As in the previous case, the operator associated with the quadratic form can be written as $Q_{a,\beta}^{-}=U_{a}^-\otimes I+\int_{[0,L)}^{\oplus} T_{a,\beta}^{-}(s)\,\D s$, where the operator $T_{a,\beta}^{-}(s)$ referring to the transverse variable acts for any $s\in [0,L)$ as $T_{a,\beta}^{-}(s)f=-f''$ with the domain
\begin{eqnarray} \label{ndomain}
\hspace{-2em} \lefteqn{\mathcal{D}(T_{a,\beta}^{-}(s))=
\big\{f\in H^2((-a,a)\setminus\{0\})\mid\: \mp\gamma_+f(\pm a)=f'(\pm a)\,,} \\ &&
f'(0_-)=f'(0_+)= -\beta^{-1}(f(0_+)-f(0_-))+\frac12 \gamma(s)(f(0_+)+f(0_-)) \big\} \nonumber
\end{eqnarray}
We are going to estimate the spectrum of $T_{a,\beta}^{-}(s)$ and to check its independence of $s$.
\begin{lemma} \label{trans-}
The operator $T_{a,\beta}^{-}(s)$ has exactly one negative eigenvalue $t_-=-\kappa_-^2$ as long as $\frac{2}{\beta}>\gamma_+$; it is independent of $s$ and for $\beta\to 0$ we have
\begin{equation*}
\kappa_-= \frac{2}{\beta} +\frac{4}{\beta}\, \frac{2-\beta\gamma_+}{2+\beta\gamma_+}\:\e^{-4a/\beta}
+\mathcal{O}\left( -\frac{4}{\beta}\, \left(\frac{2-\beta\gamma_+}{2+\beta\gamma_+}\right)^2\, \e^{-8a/\beta}\right)\,.
\end{equation*}
\end{lemma}
\begin{proof}
The function satisfying $f''(x)=\kappa^2f(x)$ for $x\ne 0$ together with the boundary conditions $\mp\gamma_+f(\pm a)=f'(\pm a)$, which has its derivative continuous at $x=0$, is of the form
\begin{equation*}
f(x) = \left\{ \begin{array}{llc} A\,\e^{\kappa x}+B\,\e^{-\kappa x} &\quad \mathrm{if}\;& x\in(-a,0) \\[.3em]
C\,\e^{\kappa x}+D\,\e^{-\kappa x} &\quad \mathrm{if}\;& x\in(0,a) \end{array} \right.
\end{equation*}
The constant $A$ is arbitrary, while for the others the requirements imply $B= AZ\,\e^{-2\kappa a}$ with $Z:= \frac{\kappa-\gamma_+}{\kappa+\gamma_+}$ and $D=-A,\: C=-B$. The remaining property from (\ref{ndomain}) leads to
\begin{equation*}
\kappa(A-B)=\frac{1}{\beta}\left(A+B-C-D\right)+\frac12\gamma(s)\left(A+B+C+D\right)\,,
\end{equation*}
and since the last term vanishes we can rewrite the spectral condition as
\begin{equation*}
\kappa =\frac{2}{\beta}\frac{1+Z\,\e^{-2\kappa a}}{1-Z\,\e^{-2\kappa a}}\,.
\end{equation*}
As before we are interested in the regime $\beta\to 0_+$. Note that as long as $Z>0$ we have $\kappa>2\beta^{-1}$, hence $\kappa$ is large and $\xi=Z\,\e^{-2\kappa a}$ is small and the expansion
\begin{equation*}
\kappa=\frac{2}{\beta}\frac{1+\xi}{1-\xi}=\frac{2}{\beta}(1+\xi)(1+\xi+\xi^2+\mathcal{O}(\xi^3))
\end{equation*}
yields the stated behaviour of $\kappa$ as $\beta\to 0_+$. The assumption $Z>0$ is satisfied for $2\beta^{-1}>\gamma_+$, and the uniqueness of the eigenvalue is a consequence of the above spectral condition and the monotonicity of the function $\kappa \mapsto \frac{1}{\kappa}\,\frac{1+Z\,\e^{-2\kappa a}}{1-Z\,\e^{-2\kappa a}}$.
\end{proof}

Next we estimate the eigenvalues of the operators $U_a^+$ and $U_a^-$ referring to the longitudinal variable $s$ in a way similar to \cite{EY}.
\begin{lemma} \label{l: longest}
There is a positive $C$ independent of $a$ and $j$ such that
\begin{equation*}
|\mu_j^\pm(a)-\mu_j|\leq Caj^2
\end{equation*}
holds for $j\in\mathbb N$ and $0<a<\frac{1}{2\gamma_+}$, where $\mu_j^\pm(a)$ are the eigenvalues of $U_a^\pm$, respectively, with the multiplicity taken into account.
\end{lemma}
\begin{proof}
We employ the operator $S_0=-\partial^2_s$ with the periodic boundary conditions, i.e. the domain $\mathcal{D}(S_0)=\{f\in L^2((0,L))\mid\, f(0)=f(L),\, f'(0)=f'(L)\}$; its eigenvalues, counting multiplicity, are $4\left[\frac{j}{2}\right]^2 \frac{\pi^2}{L^2},\: j=1,2,\dots\,,$ where $[\cdot]$ denotes as usual the entire part. Its difference from our comparison operator (\ref{comparison}) on $L^2(0,L)$ is easily estimated,
\begin{equation*}
\|S-S_0\|\leq\frac{1}{4}\gamma_+^2\,,
\end{equation*}
and consequently, by min-max principle we have
\begin{equation}
\label{p1}
\left|\mu_j-4\left[\frac{j}{2}\right]^2\frac{\pi^2}{l^2}\right|\leq\frac{1}{4}\gamma_+^2
\end{equation}
for $j\in\mathbb{N}$. Next we can use another simple estimate,
\begin{equation*}
U_a^+-\frac{1}{(1-a\gamma_+)^2}S=\frac{a\gamma''_+}{2(1-a\gamma_+)^3}
-\frac{\gamma^2}{4(1+a\gamma_+)^2}+\frac{\gamma^2}{4(1-a\gamma_+)^2}\,,
\end{equation*}
and since the last two terms equal $a\gamma_+ \gamma^2 (1-a^2\gamma_+^2)^{-2}$, we infer that
\begin{equation}
\label{p2} \left|\mu_j^+-\frac{\mu_j}{(1-a\gamma_+)^2}\right|\leq c_0a
\end{equation}
holds for some $c_0>0$ and any $j\in\mathbb{N}$. Combining now (\ref{p1}) and (\ref{p2}) we get
\begin{eqnarray*}
\lefteqn{|\mu_j^+-\mu_j| \leq\left|\mu_j^+-\frac{\mu_j}{(1-a\gamma_+)^2}\right|
+|\mu_j|\cdot\left|\frac{1-(1-a\gamma_+)^2}{(1-a\gamma_+)^2}\right|} \\[.3em] && \quad \leq
c_0 a +c_1 a|\mu_j|\leq C a j^2
\end{eqnarray*}
with suitable constants. The second inequality is checked in a similar way: we use
\begin{equation*}
\hspace{-5em} U_a^--\frac{1}{(1+a\gamma_+)^2}S =-\frac{a\gamma''_+}{2(1-a\gamma_+)^3}
-\frac{5a^2(\gamma'_+)^2}{4(1-a\gamma_+)^4}
-\frac{a\gamma_+}{(1-a\gamma_+)^2(1+a\gamma_+)^2}\gamma^2
\end{equation*}
which implies
\begin{equation*}
\left|U_a^--\frac{1}{(1+a\gamma_+)^2}S\right| \leq c_0 a+c_1 a^2\leq c_2 a\,,
\end{equation*}
where in the second inequality we employed the fact that $a$ is bounded. With help of min-max principle we then get
\begin{equation*}
\left|\mu_j^--\frac{\mu_j}{(1+a\gamma_+)^2}\right| \leq c_2a
\end{equation*}
hence finally we arrive at the inequality
\begin{equation*}
|\mu_j^--\mu_j| \leq
c_2a +|\mu_j|\,\left|\frac{1-(1+a\gamma_+)^2}{(1+a\gamma_+)^2}\right|
\leq c_2a +c_3 a |\mu_j|\leq C a j^2
\end{equation*}
valid for a suitable $C$ which completes the proof.
\end{proof}

Now we are ready to prove our first main result:

\medskip
\noindent We define $a(\beta)=-\frac{3}{4}\beta\ln\beta$ and denote the eigenvalues of the operators $T_{a(\beta),\beta}^{\pm}$ as $t_{\pm,\beta}^j$, respectively, their multiplicities being taken into account. From Lemmata~\ref{trans+} and \ref{trans-} we know that $t_{\pm,\beta}^1=t_\pm$ for $\beta$ small enough, while $t_{\pm,\beta}^j\geq 0$ holds for $j>1$. Collecting the estimates worked out above we have
\begin{eqnarray}
\lefteqn{Q_{a(\beta),\beta}^{-} = U_{a(\beta)}^-\otimes I+\int_{(0,L)}^\oplus
T_{a(\beta),\beta}^{-}(s)\,D s \leq H_N(\beta) \leq H_\beta} \nonumber \\[.3em] &&
\qquad \leq U_{a(\beta)}^+\otimes I+\int_{(0,L)}^\oplus
T_{a(\beta),\beta}^{+}(s)\,\D s =Q_{a(\beta),\beta}^{+} \label{fullest}
\end{eqnarray}
and the eigenvalues of the operators $Q_{a(\beta),\beta}^{\pm}$ between which we squeeze our singular Schr\"odinger operator $H_\beta$ are naturally $t_{\pm,\beta}^k +\mu^\pm_j(a(\beta))$ with $k,j\in\mathbb{N}$. Those with $k\geq2$ and $j\in\mathbb{N}$ are uniformly bounded from below in view of the inequality
\begin{equation}
\label{l1}t_{\pm,\beta}^k+\mu^\pm_j(a(\beta))\geq\mu^{\pm}_1(a(\beta))=\mu_1+\mathcal{O}(-\beta\ln(\beta))\,,
\end{equation}
hence we can focus on $k=1$ only. For $j\in\mathbb N$ we denote
\begin{equation*}
\omega_{\pm,\beta}^j=t_{\pm,\beta}^1+\mu^\pm_j(a(\beta))
\end{equation*}
With our choice of $a(\beta)$ we have $\mathrm{e}^{-4\kappa a}= \beta^3$ so from the above lemmata we get $\kappa_\pm = \frac{2}{\beta} + \mathcal{O}(\beta)$ and $\mu^\pm_j(a(\beta))$ differ from $\mu_j$ by $\mathcal{O}(-\beta j^2|\ln\beta|)$; putting these estimates together we can conclude that
\begin{equation}
\label{l2} \omega_{\pm,\beta}^j=-\frac{4}{\beta^2}+\mu_j+\mathcal
O(-\beta\ln\beta)\quad\textrm{as}\;\;\beta\rightarrow0_+
\end{equation}
with the error term in general dependent on $j$. Combining (\ref{l1}) and (\ref{l2}) we can conclude that to any $n\in \mathbb{N}$ there is a $\beta(n)>0$ such that
\begin{equation*}
\omega_{+,\beta}^n\leq 0\,,\quad \omega_{+,\beta}^n<t_{+,\beta}^k+\mu^+_j(a(\beta))
\quad\mathrm{and}\quad
\omega_{-,\beta}^n< t_{-,\beta}^k+\mu^-_j(a(\beta))
\end{equation*}
holds for $\beta\leq\beta(n)$, $k\geq2$, and $j\geq1$. Hence $j$-th eigenvalue of $Q_{a(\beta),\beta}^{\pm}$, counting multiplicity, is $\omega_{\pm,\beta}^j$ for all $j\leq n$ and $\beta\leq\beta(n)$. Furthermore, for $\beta\leq\beta(n)$ we denote $\xi_{+}^j(\beta)$ and $\xi_{-}^j(\beta)$ the $j$-th eigenvalue of $H_D(\beta)$ and $H_N(\beta)$, respectively; then from (\ref{fullest}) and the min-max principle we obtain
\begin{equation*}
\omega_{-,\beta}^n\leq\xi_{-}^j(\beta)\qquad\xi_{+}^j(\beta)\leq\omega_{+,\beta}^n
\end{equation*}
for $j=1,2,\dots,n\,$, which in particular implies $\xi_{+}^n(\beta)<0$. Using the min-max principle once again we conclude that $H_\beta$ has at least $n$ eigenvalues in the interval $(-\infty,\xi_{+}^n(\beta))$ and for any $1\leq j\leq n$ we have $\xi_{-}^j(\beta)\leq \lambda_j\leq \xi_{+}^j(\beta)$ which completes the proof.

\section{Proof of Theorem~\ref{thm2}}
\setcounter{equation}{0}

For a self-adjoint operator $A$ with $\inf\sigma_\mathrm{ess}(A)=0$ we put $N^-(A):=\#\{\sigma_d(A)\cap(-\infty,0)\}$. In view of (\ref{fullest}) the eigenvalue number of $H_\beta$ can be estimated as
\begin{equation}
\label{ner} N^-(Q_{a,\beta}^{-})\leq
N^-(H_N\beta)\leq\#\sigma_d(H_\beta)\leq N^-(H_D(\beta))\leq
N^-(Q_{a,\beta}^{+})
\end{equation}
In order to use this estimate we define
\begin{equation*}
K^\pm_\beta=\{j\in\mathbb N|\:\omega_{\pm,\beta}^j<0\}
\end{equation*}
and derive the following asymptotic expansions of these quantities.

\begin{lemma} \label{strongK}
In the strong coupling limit, $\beta\rightarrow0_+$, we have
\begin{equation}
\#K^\pm_\beta=\frac{2L}{\pi\beta} +\mathcal{O}(|\ln\beta|)\,.
\end{equation}
\end{lemma}
\begin{proof}
We choose $K$ such that $\beta^{-1}>K>0$ and $(\beta^{-1}-K)^2<\beta^{-2}-4\beta-16^{-1}\gamma_+^2$. With the preceding proof in mind we can write
\begin{equation*}
K^+_\beta=\{j\in\mathbb N\,|\:t_{+,\beta}^1+\mu^+_j(a(\beta))<0\}\,.
\end{equation*}
Lemma~\ref{trans+} allows us to make the following estimate,
\begin{equation*}
K^+_\beta \supset \left\{j\in\mathbb{N}\,\big|\: \mu_j+Ca(\beta)j^2
<\frac{4}{\beta^2}-\frac{16}{\beta^2}\,\e^{-4a(\beta)/\beta}
=\frac{4}{\beta^2}-16\beta\right\}\,;
\end{equation*}
using further (\ref{p1}) and the indicated choice of $K$ we infer that
\begin{eqnarray*}
K^+_\beta\supset \left\{j\in\mathbb{N}\,\Big|\: 4\left[\frac{j}{2}\right]^2\frac{\pi^2}{L^2}+Ca(\beta)j^2
<\frac{4}{\beta^2}-16\beta-\frac{1}{4}\gamma_+^2 \right\}\\
\quad\;\;\; \supset \left\{j\in\mathbb{N}\, \Big|\: j^2\frac{\pi^2}{L^2}-\frac34 C\,\beta\ln\beta\,j^2 <4\left(\frac{1}{\beta}-K\right)^2 \right\}\\
\quad\;\;\; \supset \left\{j\in\mathbb{N}\,\Big|\: j<2\left(\frac{1}{\beta}-K\right)\left(\frac{\pi^2}{L^2}-\frac34 C \beta \ln\beta\right)^{-1/2} \right\}
\end{eqnarray*}
We employ the Taylor expansion $(M+x)^{-1/2}=M^{-1/2}-\frac12\,xM^{-3/2}+\mathcal{O}(x^2)$; since we are interested in the asymptotics $\beta\rightarrow0_+$, we rewrite the right-hand side of the last inequality as
\begin{equation*}
2\left(\frac{1}{\beta}-K\right)\left(\frac{\pi^2}{L^2}-\frac34 C\,\ln\beta\right)^{-\frac{1}{2}}
\simeq 2\left(\frac{1}{\beta}-K\right) \left[\frac{L}{\pi}+\frac38 C\,\beta\ln\beta \left(\frac{L}{\pi}\right)^3\right]\,,
\end{equation*}
which allows us to infer that
\begin{equation}
\label{k1}\#K^+_\beta\geq \frac{2L}{\pi\beta}+\mathcal{O}(|\ln\beta|)
\end{equation}
holds as $\beta\rightarrow0_+$. In a similar way we estimate
$\#K^-_\beta$. First we choose a number $K'$ satisfying $0<K'<
\big(4\beta+\frac{\gamma_+^2}{16}\big)^{1/2}$ and note that
$\frac{1}{\beta^2} +4\beta+\frac{\gamma_+^2}{16}
<\left(\frac{1}{\beta}+K'\right)^2$. Then we have
\begin{eqnarray*}
K^-_\beta= \left\{j\in\mathbb{N}\,\big|\: t_{-,\beta}^1+\mu^-_j(a(\beta))<0 \right\}\\
\quad\;\; \subset \left\{j\in\mathbb{N}\,\Big|\: \mu_j-Ca(\beta)j^2<\frac{4}{\beta^2}+\frac{16}{\beta^2}\,
\frac{2-\beta\gamma_+}{2+\beta\gamma_+}\,\e^{-4a(\beta)/\beta} \right \}\\
\quad\;\; \subset \left\{j\in\mathbb{N}\,\Big|\: \mu_j+\frac{3}{4}C\beta\ln\beta\,j^2 <\frac{4}{\beta^2}+16\beta\,\frac{2-\beta\gamma_+}{2+\beta\gamma_+} \right\}
\end{eqnarray*}
With help of the fact that $2(j-1)\geq j$ for $j>1$ we further have
\begin{eqnarray*}
\hspace{-1.8em} K^-_\beta\subset \{1\}\cup \left\{j\geq2\,\Big|\: \left(\frac{(j-1)\pi}{L}\right)^2+\frac{3}{4}C\beta\ln\beta\,(j-1)^2<
\frac{4}{\beta^2}+16\beta+\frac{\gamma_+^2}{4} \right\}\\
\subset\{1\}\cup \left\{j\geq2\,\Big|\: (j-1)^2<
\left(\frac{4}{\beta^2}+16\beta+\frac{\gamma_+^2}{4}\right) \left(\left(\frac{\pi}{L}\right)^2+\frac{3}{4}C\beta\ln\beta\right)^{-1}
\right\}\\
\subset\{1\}\cup \left\{j\geq2\,\Big|\: j<1+
2\left(\frac{1}{\beta}+K'\right)\left(\left(\frac{\pi}{L}\right)^2+\frac{3}{4}C\beta\ln\beta
\right)^{-1/2} \right\}
\end{eqnarray*}
Now we can estimate the expression on the right-hand side of the last inequality in the asymptotic regime $\beta\rightarrow0_+$ as
\begin{equation*}
2\left(\frac{1}{\beta}+K'\right)\left(\left(\frac{\pi}{L}\right)^2
+\frac{3}{4}C\beta\ln\beta\right)^{-\frac{1}{2}}\simeq\frac{2L}{\pi\beta}+\mathcal{O}(|\ln\beta|)\,.
\end{equation*}
In combination with the above inclusions this leads to
\begin{equation}
\label{k2} \#K^-_\beta\leq\frac{2l}{\pi\beta}
+\mathcal{O}(|\ln\beta|)
\end{equation}
as $\beta\rightarrow0_+$. Finally, we know that $t_{+,\beta}^1<t_{-,\beta}^1$ which implies $K^+_\beta\subset K^-_\beta$, and this together with (\ref{k1}) and (\ref{k2}) concludes the proof.
\end{proof}

We also need to estimate the second eigenvalue of the operators $T_{a(\beta),\beta}^{-}(s)$.

\begin{lemma}
$T_{a,\beta}^{-}(s)$ with a fixed $s\in[0,L)$ has no eigenvalues in $\Big[0,\min\Big\{\frac{\gamma_+}{2a}, \big(\frac{\pi}{4a}\big)^2\Big\}\Big)$ provided $0<\beta<2a$.
\end{lemma}
\begin{proof}
Let us check first that zero is not an eigenvalue. The corresponding eigenfunction should have to be linear and the conditions $\mp\gamma_+f(\pm a)=f'(\pm a)$ and $f'(0_-)=f'(0_+)$ would require $f(x)=\pm A(\mp\gamma_+x+1+\gamma_+a)$ for $\pm x\in(0,a)$, and as in Lemma~\ref{trans-} the spectral condition would read $-\gamma_+=\frac{2}{\beta}(1+\gamma_+a)$ which cannot be true because the right-hand side is positive. Furthermore, the spectral condition for an eigenvalue $k^2>0$ is found again as in Lemma~\ref{trans-}; after s simple calculation we find that it reads
\begin{equation*}
\frac12 \beta =\frac{1}{k}\,\frac{\gamma_+\tan ka +k}{\gamma_+ -k\tan ka}
\end{equation*}
The right-hand side can be estimated by $\frac{1+\gamma_+a}{\gamma_+-k^2a}$ provided that $ka<\frac{\pi}{2}$ and
at the same time $\gamma_+-k\tan ka >0$; finding the value for which this expression equals $\frac12\beta$ we would obviously get a lower bound to $k$. Rewriting the condition as
\begin{equation*}
-ka^2=(1+\gamma_+a)\frac{2}{\beta}-\gamma_+
\end{equation*}
we see that the left-hand side is negative while the right-hand side is positive under our assumption, hence one has to ask about the restriction coming from the condition $\gamma_+-k\tan ka >0$. In particular, for $ka< \frac14\pi$ this is true provided $\gamma_+-2k^2a>0$, which means that the spectral problem has no solution is $k^2$ is smaller either than $\frac{\gamma_+}{2a}$ or $\left(\frac{\pi}{4a}\right)^2$ which concludes the argument.
\end{proof}

Now we are ready to prove our second main result:

\medskip

\noindent We begin by showing that the relation
\begin{equation}
\label{rov} N^-(Q_{a(\beta),\beta}^{-})=\#K_\beta^-
\end{equation}
holds for any sufficiently small $\beta>0$. We know that all the eigenvalues of $Q_{a(\beta),\beta}^{-}$ can be written as $\{t_{-,\beta}^j +\mu^-_k(a(\beta))\}_{j,k\in\mathbb N}$ with the multiplicity taken into account. From the previous lemma we have $t_{-,\beta}^2>\textrm{min}\left\{\frac{\gamma_+}{2a},\left(\frac{\pi}{4a}\right)^2\right\}$ which together with $|\mu_j^-(a)-\mu_j|\leq Caj^2$ implies existence of a $\beta_0$ such that
\begin{equation*}
t_{-,\beta}^k+\mu^-_j(a(\beta))>0
\end{equation*}
holds for $j>1$, $\,k\geq1$, and $\beta\in(0,\beta_0)$. This implies
\begin{eqnarray*}
\lefteqn{N^-(Q_{a(\beta),\beta}^{-})=\#\{(k,j)\in\mathbb{N}^2\,|\:
t_{-,\beta}^k+\mu^-_j(a(\beta))<0\}} \\[.3em]
&& \quad = \#\{j\in\mathbb{N}\,|\:t_{-,\beta}^1+\mu^-_j(a(\beta))<0\} =: K_\beta^-\,,
\end{eqnarray*}
i.e. the relation (\ref{rov}); combining it with (\ref{ner}) we obtain
\begin{equation*}
\#K_\beta^+\leq\#\sigma_d(H_\beta)\leq
N^-(Q_{a,\beta}^{-})=\#K_\beta^-
\end{equation*}
which by virtue of Lemma~\ref{strongK} concludes the proof.

\section{Concluding remarks}

We have seen that, despite very different eigenfunctions, the $\delta'$ `leaky loops' behave in the strong-coupling regime similarly to their $\delta$ counterparts: the number of negative eigenvalues is given in the leading order by a Weyl-type term, and the eigenvalues themselves are after a natural renormalization determined by the one-dimensional Schr\"odinger operator with the known curvature-induced potential.

The question is whether and how the current results can be extended. The bracketing technique we used would work for infinite smooth curves $\Gamma$ without ends provided we impose suitable regularity assumptions. If, on the other hand, the curve is finite or semi-infinite the situation becomes more complicated because one has to impose appropriate boundary conditions at the endpoints of the interval on which the comparison operator (\ref{comparison}) is defined. One can modify the present argument to get an estimate on the number of eigenvalues because there those boundary conditions play no role, the counting functions in the Dirichlet and Neumann case differing by an $\OO(1)$ term. For an eigenvalue position estimate, on the other hand, this is not sufficient and one conjectures that the \emph{Dirichlet} comparison operator has to be used. For a two-dimensional open arc $\Gamma$ supporting a $\delta$ interaction this conjecture has been proved recently \cite{EP}; the argument is more complicated because one cannot use operators with separating variables. We believe that the same method could work in the $\delta'$ case too, however, the question is not simple and we postpone discussing it to another paper.

On the other hand, finding the asymptotics in the case when $\Gamma$ is not smooth, or even has branching points, represents a much harder problem and the answer is not known even in the $\delta$ case, although some inspiration can be found in squeezing limits of Dirichlet tubes --- see, e.g., \cite{CE}.

\section*{Acknowledgments}
The research was supported by the Czech Science Foundation  within
the project P203/11/0701 and by by Grant Agency of the Czech
Technical University in Prague, grant No. SGS13/217/OHK4/3T/14.

\section*{References}
\bibliographystyle{alpha}

\end{document}